\documentclass{article}
\usepackage{fullpage} 

\usepackage{mdwlist}

\usepackage{amssymb,amsbsy,latexsym}
\usepackage{amsmath}
\usepackage{graphics, subfigure, float}
\usepackage{fp, calc}
\usepackage{bm}
\usepackage{hyperref}

\usepackage[latin1]{}
\usepackage[english]{babel} 
\usepackage[T1]{fontenc} 

\usepackage{hyperref}

\usepackage{bm}

\usepackage{amscd,amsthm}
\usepackage{datetime}
\usepackage{verbatim, comment}

\usepackage[numbers]{natbib}

\usepackage{pst-all}

\newtheoremstyle{theorem}{1em}{1em}{\slshape}{0pt}{\bfseries}{.}{ }{}
\theoremstyle{theorem}
\newtheorem{theorem}{Theorem}
\newtheorem*{theorem*}{Theorem}
\newtheorem{corollary}[theorem]{Corollary}

\newtheorem{lemma}[theorem]{Lemma}

\theoremstyle{remark}

\newtheorem*{remark*}{Remark}

\providecommand{\setN}{\mathbb{N}}

\providecommand{\setR}{\mathbb{R}}

\theoremstyle{theorem}
\theoremstyle{definition}

\newtheorem*{claim*}{Claim}

\usepackage[displaymath,textmath,graphics, subfigure, floats]{preview} 
\PreviewEnvironment{myproblem}
\PreviewEnvironment{defn} 
\PreviewEnvironment{center} 
\PreviewEnvironment{pspicture} 
\usepackage{datetime}

\makeatother

\begin{document}

\title{Deterministic Discrepancy Minimization via the Multiplicative Weight Update Method}
\date{University of Washington, Seattle} 
\author{Avi Levy\thanks{Email: {\tt avius@uw.edu}} \and Harishchandra Ramadas\thanks{Email: {\tt ramadas@math.washington.edu}} \and Thomas Rothvoss\thanks{Email: {\tt rothvoss@uw.edu}. Supported by NSF grant 1420180 with title ``\emph{Limitations of convex relaxations in combinatorial optimization}'',  an Alfred P. Sloan Research Fellowship and a David \& Lucile Packard Foundation Fellowship. File compiled on {\today, \currenttime}.}} 

\maketitle

\begin{abstract}
A well-known theorem of Spencer shows that any set system with $n$ sets over $n$ elements
admits a coloring of discrepancy $O(\sqrt{n})$. While the original proof was non-constructive, 
recent progress brought polynomial time algorithms by Bansal, Lovett and Meka, and Rothvoss. 
All those algorithms are randomized, even though Bansal's algorithm admitted a complicated
derandomization. 

We propose an elegant deterministic polynomial time algorithm that is inspired by Lovett-Meka as well as the Multiplicative Weight Update method. The algorithm iteratively updates a fractional coloring while 
controlling the exponential weights that are assigned to the set constraints.

A conjecture by Meka suggests that Spencer's bound can be generalized to symmetric matrices. We 
prove that $n \times n$ matrices that are block diagonal with block size $q$ admit a coloring
of discrepancy $O(\sqrt{n} \cdot \sqrt{\log(q)})$. 

Bansal, Dadush and Garg recently gave a randomized algorithm to find a vector $x$ with entries in $\lbrace{-1,1\rbrace}$ with $\|Ax\|_{\infty} \leq O(\sqrt{\log n})$ in polynomial time, where $A$ is any matrix whose columns have length at most 1. We show that our method can be used to deterministically obtain such a vector.
\end{abstract}

\section{Introduction}

The classical setting in (combinatorial) \emph{discrepancy theory} is that a 
set system $S_1,\ldots,S_m \subseteq \{ 1,\ldots,n\}$  over a ground set of $n$ elements is
given and the goal is to find \emph{bi-coloring} $\chi : \{ 1,\ldots,n\} \to \{ \pm 1\}$ 
so that the worst imbalance $\max_{i=1,\ldots,m} |\chi(S_i)|$ of a set
is minimized. Here we abbreviate $\chi(S_i) := \sum_{j \in S_i} \chi(j)$.
A seminal result of Spencer~\cite{SixStandardDeviationsSuffice-Spencer1985} 
says that there is always a coloring 
$\chi$ where the imbalance is at most $O(\sqrt{n \cdot \log(2m/n)})$ for $m \geq n$. 
The proof of Spencer is based on the \emph{partial coloring method} that was 
first used by Beck in 1981~\cite{Beck-RothsEstimateIsSharp1981}. The argument applies the
 \emph{pigeonhole principle} to obtain that many of the $2^n$ many colorings $\chi,\chi'$ must satisfy 
 $|\chi(S_i) - \chi'(S_i)| \leq O(\sqrt{n \cdot \log(2m/n)})$ for all sets $S_i$. Then one can take the \emph{difference}
between such a pair of colorings with $|\{ j \mid \chi(j) \neq \chi'(j) \}| \geq \frac{n}{2}$ to obtain a \emph{partial coloring} of low discrepancy. This partial coloring can be used to color half of the elements. Then one iterates the argument and again finds a partial coloring. As the remaining set system has only half the elements, the bound in the second iteration becomes better by a constant factor. This process is repeated until all elements are colored; the total discrepancy is then given by a convergent series with value $O(\sqrt{n \cdot \log(2m/n)})$.
More general arguments based on convex geometry were given by Gluskin~\cite{RootNDisc-Gluskin89}
and by Giannopoulos~\cite{Giannopoulos1997}, but their arguments still relied on a pigeonhole principle
with exponentially many pigeons and pigeonholes and did not lead to polynomial time algorithms. 

In fact, Alon and Spencer~\cite{ProbabilisticMethod-AlonSpencer08} even conjectured that finding a coloring satisfying Spencer's theorem
would by intractable.  
In a breakthrough, Bansal~\cite{DiscrepancyMinimization-Bansal-FOCS2010} showed that one could set up
a \emph{semi-definite program} (SDP) to find at least a vector coloring, using Spencer's Theorem to argue that the SDP has to be feasible. He then argued that a random walk guided by updated solutions to that SDP would find a coloring of discrepancy $O(\sqrt{n})$ in the balanced case $m=n$.
 However, his approach needed a very careful choice of parameters.
 
A simpler and truly constructive approach that does not rely on Spencer's argument
was provided by Lovett and Meka~\cite{DiscrepancyMinimization-LovettMekaFOCS12}, who showed that for $x^{(0)} \in [-1,1]^n$, any  polytope of the form 
$P = \{ x \in [-1,1]^n : \left|\left<v_i,x-x^{(0)}\right>\right| \leq \lambda_i \; \forall i \in [m] \}$ 
contains a point that has at least half of the coordinates in $\{-1,1\}$. 
Here it is important that the polytope $P$ is large enough; if the normal vectors $v_i$ are scaled to unit length, then the argument requires that $\sum_{i=1}^m e^{-\lambda_i^2/16} \leq \frac{n}{16}$ holds. Their algorithm surprisingly simple: start a Brownian motion at $x^{(0)}$ and stay inside any face that is hit at any time. They showed that this random walk eventually reaches a point with the desired properties. 

More recently, the third author provided another algorithm which simply consists of taking a random Gaussian vector $x$ 
and then computing the nearest point to $x$ in $P$. In contrast to both of the previous algorithms, this argument
extends to the case that $P = Q \cap [-1,1]^n$ where $Q$ is any symmetric convex set with a large enough Gaussian measure.

However, all three algorithms described above are randomized, although Bansal and 
Spencer~\cite{DeterministicDiscrepancyMinimizationBansalSpencerJournal13} could derandomize 
the original arguments by Bansal. They showed that the random walk already works 
if the directions are chosen from a 4-wise independent distribution, which then allows a 
polynomial time derandomization. 

In our algorithm, we think of the process more as a \emph{multiplicative weight update} procedure, 
where each constraint has a weight that increases if the current point moves in the direction of
its normal vector. The potential function we consider is the sum of those weights. 
Then in each step we simply need to select an update direction in which the potential function
does not increase.

The multiplicative weight update method is a meta-algorithm that originated in game theory but 
has found numerous recent applications in theoretical computer science and machine learning. In the general setting one imagines having a set of experts (in our case the set constraints) that are 
assigned an exponential weight that reflects the value of the gain/loss that expert's decisions had 
in previous rounds. Then in each iteration one selects an update, which can be a convex combination of experts, where the convex coefficient is proportional to the current weight of the expert\footnote{We should mention for the sake of completeness that our update choice is \emph{not} a convex combination of the experts weighted by their exponential weights.}. 
We refer to the very readable 
survey of Arora, Hazan and Kale~\cite{MWU-Survey-Arora-HazanKale2012} for a detailed discussion.

\subsection{Related work}

If we have a set system $S_1,\ldots,S_m$ where each element lies in 
at most $t$ sets, then the partial coloring
technique described above can be used to  find a coloring of discrepancy $O(\sqrt{t} \cdot \log n)$~\cite{DiscrepancyBound-sqrtT-logN-SrinivasanSODA97}. A linear programming approach of
Beck and Fiala~\cite{IntegerMakingTheorems-BeckFiala81} showed that the discrepancy
is bounded by $2t-1$, independent of the size of the set system. 
On the other hand, there is a non-constructive approach of 
Banaszczyk~\cite{BalancingVectors-Banaszczyk98} that provides a bound of $O(\sqrt{t \log n})$ 
using convex geometry arguments. Only very recently, a corresponding algorithmic
bound was found by Bansal, Dadush and Garg~\cite{ConstructiveBanaBansalDadushGarg16}.
A conjecture of Beck and Fiala says that 
the correct bound should be $O(\sqrt{t})$. This bound can be achieved 
for the vector coloring version, see Nikolov~\cite{NikolovVectorColoringKomlosArxiv2013}.

More generally, the theorem of 
Banaszczyk~\cite{BalancingVectors-Banaszczyk98} shows that for any convex set $K$ with Gaussian measure at least $\frac{1}{2}$
and any set of vectors $v_1,\ldots,v_m$ of length $\|v_i\|_2 \leq \frac{1}{5}$, there exist 
signs $\varepsilon_i \in \{ \pm 1\}$ so that $\sum_{i=1}^m \varepsilon_iv_i \in K$.

A set of $k$ permutations on $n$ symbols induces a set system with $kn$ sets
given by the prefix intervals. One can use the partial coloring method
to find a  $O(\sqrt{k} \log n)$ discrepancy coloring~\cite{DiscrepancyOfPermutations-SpencerEtAl}, 
while a linear programming
approach gives a $O(k \log n)$ discrepancy~\cite{DiscrepancyOf3Permutations-Bohus90}. 
In fact, for any $k$ one can always color half of the elements with a
discrepancy of $O(\sqrt{k})$ --- this even holds for each induced sub-system~\cite{DiscrepancyOfPermutations-SpencerEtAl}.
Still, \cite{CounterexampleBecksPermutationConjecture-FOCS12} constructed  
3 permutations requiring a discrepancy of 
$\Theta(\log n)$ to color all elements. 

Also the recent proof of the Kadison-Singer conjecture by Marcus, Spielman
and Srivastava~\cite{KadisonSingerConjecture-MarcusSpielmanSrivastavaArxiv2013} can be 
seen as a discrepancy result. 
They show that a set of vectors $v_1,\ldots,v_m \in \setR^n$ with $\sum_{i=1}^m v_iv_i^T = I$
can be partitioned into two halves $S_1,S_2$ so that $\sum_{i \in S_j} v_iv_i^T \preceq (\frac{1}{2} + O(\sqrt{\varepsilon}))I$ for $j \in \{ 1,2\}$
where $\varepsilon = \max_{i=1,\ldots,m}\{ \|v_i\|_2^2 \}$
and $I$ is the $n \times n$ identity matrix. Their method is based on
interlacing polynomials; no polynomial time algorithm is known to find
the desired partition.

For a symmetric matrix $A \in \setR^{m \times m}$, let $\|A\|_{\textrm{op}}$ denote the largest singular value; in other words, the largest absolute value of any eigenvalue. 
The discrepancy question can be generalized from sets to symmetric matrices $A_1,\ldots,A_n \in \setR^{m \times m}$
with $\|A_i\|_{\textrm{op}} \leq 1$ by defining $\textrm{disc}(\{A_1,\ldots,A_n\}) := \min\{ \|\sum_{i=1}^n x_iA_i\|_{\textrm{op}} : x \in \{ -1,1\}^n\}$. Note that picking 0/1 diagonal matrices $A_i$ corresponding to the incidence vector of element $i$
would exactly encode the set coloring setting.
Again the interesting case is $m = n$; in contrast to the diagonal case it is only known that the
discrepancy is bounded by $O(\sqrt{n \cdot \log(n)})$, which is already attained by a random coloring. 
Meka\footnote{See the blog post\newline \url{https://windowsontheory.org/2014/02/07/discrepancy-and-beating-the-union-bound/}.} conjectured that the discrepancy of $n$ matrices can be bounded by $O(\sqrt{n})$.

For a very readable introduction into discrepancy theory, we recommend
Chapter~4 in the book of Matou{\v s}ek~\cite{GeometricDiscrepancy-Matousek99}
or the book of Chazelle~\cite{DiscrepancyMethod-Chazelle2001}.

\subsection{Our contribution}

Our main result is a deterministic version of the theorem of Lovett and Meka:
\begin{theorem} \label{thm:DeterministicLovettMeka}
Let $v_1,\ldots,v_m \in \setR^n$ unit vectors, $x^{(0)} \in [-1,1]^n$ be a starting point and let $\lambda_1\geq \ldots\geq\lambda_m \geq 0$ be parameters
so that $\sum_{i=1}^m \exp(-\lambda_i^2/16) \leq \frac{n}{32}$. Then there is a deterministic  algorithm that computes a vector $x \in [-1,1]^n$ with
$\left<v_i,x-x^{(0)}\right> \leq 8\lambda_i$ for all $i \in [m]$ and $|\{ i : x_i =\pm 1 \}| \geq \frac{n}{2}$, in time $O(\min\{ n^4m, n^3m\lambda_1^2\})$. 
\end{theorem}
By setting $\lambda_i=O(1)$ this yields a deterministic version of Spencer's theorem in the balanced case $m=n$:
\begin{corollary} \label{lem:RunningTimeSpencersTheorem}
Given $n$ sets over $n$ elements, there is a deterministic algorithm that finds a $O(\sqrt{n})$-discrepancy
coloring in time $O(n^4)$.
\end{corollary}
Furthermore, Spencer's \emph{hyperbolic cosine algorithm}~\cite{BalancingGamesSpencerJCTB77} can also be interpreted
as a multiplicative weight update argument. However, the techniques of \cite{BalancingGamesSpencerJCTB77}
are only enough for a $O(\sqrt{n \log(n)})$ discrepancy bound for the balanced case. 
Our hope is that similar arguments can be applied to solve open problems such as whether 
there is an extension of Spencer's result to balance matrices~\cite{MatrixBalancingZouziasICALP12} 
and to better discrepancy minimization techniques in the Beck-Fiala setting. 
To demonstrate the versatility of our arguments, we show an extension to the matrix discrepancy case.

We say that a symmetric matrix $A \in \setR^{m \times m}$ is \emph{$q$-block diagonal} if
it can be written as $A = \textrm{diag}(B_1,\ldots,B_{m/q})$, where each $B_j$ is a symmetric $q \times q$
matrix. 
\begin{theorem} \label{thm:MatrixBalancing}
For given $q$-block diagonal matrices $A_1,\ldots,A_n \in \setR^{m \times m}$ with $\|A_i\|_{\textrm{op}} \leq 1$ for $i=1,\ldots,n$ one can compute a coloring $x \in \{ -1,1\}^n$ with $\|\sum_{i=1}^n x_iA_i\|_{\textrm{op}} \leq O(\sqrt{n \log(\frac{2qm}{n})})$ deterministically in time $O(n^5 + n^4m^3)$. 
\end{theorem}
Finally, we can also give the first deterministic algorithm for the result of Bansal, Dadush and Garg~\cite{ConstructiveBanaBansalDadushGarg16}. 
\begin{theorem} \label{thm:ConstructiveBeckFialaAlgorithm}
Let $A \in \setR^{m \times n}$ be a matrix with $\|A^j\|_2 \leq 1$ for all columns $j=1,\ldots,n$. Then
there is a deterministic algorithm to find a coloring $x \in \{ -1,1\}^n$ with $\|Ax\|_{\infty} \leq O(\sqrt{\log n})$
in time $O(n^3\log(n) \cdot (m+n))$.
\end{theorem}
While \cite{ConstructiveBanaBansalDadushGarg16} need to solve a semidefinite program in each step of their random walk, our algorithm does not require solving any SDPs. Note that we do not optimize running times such as by using fast matrix multiplication.

In the Beck-Fiala setting, we are given a set system over $n$ elements, where each element is contained in at most $t$ subsets. Theorem \ref{thm:ConstructiveBeckFialaAlgorithm} then provides the first polynomial-time deterministic algorithm that produces a coloring with discrepancy $O(\sqrt{t\log n})$; we simply choose the matrix $A$ whose rows are the incidence vectors of members of the set system, scaled by $1/\sqrt{t}$. 

For space reasons, we defer the proof of Theorem \ref{thm:MatrixBalancing} to Appendix \ref{BoundedColumns}.

\section{The algorithm for partial coloring}

We will now describe the algorithm proving Theorem~\ref{thm:DeterministicLovettMeka}.  First note that for any $\lambda_i>2\sqrt{n}$ we can remove the constraint $\left<v_i,x - x_0\right> \leq \lambda_i$, as it does not cut off any point in $[-1,1]^n$. Thus we assume without loss of generality that $2\sqrt{n}\geq \lambda_1\geq \cdots \geq \lambda_m\geq 0$. Let $\delta:=\frac{1}{\lambda_1}$ denote the step size of our algorithm. The algorithm will run for $O(n/\delta^2)$ iterations, each of computational cost $O(n^2m)$. Note that $\delta=O(1/\sqrt{n})$ so the algorithm terminates in $O(n^2)$ iterations. The total runtime is hence $O(n^2m\cdot n/\delta^2) = O(n^3m\lambda_1^2)\leq O(n^4m)$.

For a symmetric matrix $M \in \setR^{n \times n}$ we know that an \emph{eigendecomposition} 
$M = \sum_{j=1}^{n} \mu_j u_ju_j^T$ can be computed in time $O(n^3)$. 
Here $\mu_j := \mu_j(M)$ is the \emph{$j$th eigenvalue} of $M$ and $u_j := u_j(M)$ is the 
corresponding \emph{eigenvector} with $\|u_j\|_2 = 1$.
We make the convention that the eigenvalues are sorted as $\mu_1 \geq \ldots \geq \mu_n$.
The algorithm is as follows:
\begin{enumerate}
\item[(1)] Set weights $w_i^{(0)} = \exp(-\lambda_i^2)$ for all $i=1,\ldots,m$. 
\item[(2)] FOR $t = 0$ TO $\infty$ DO
  \begin{enumerate}
  \item[(3)] Define the following subspaces
\begin{itemize*}
\item $U_1^{(t)} := \textrm{span}\{ e_j : -1 < x^{(t)}_j < 1\}$
\item $U_2^{(t)} := \{ x \in \setR^n\mid \big<x,x^{(t)}\big> = 0\}$
\item $U_3^{(t)} := \{ x \in \setR^n \mid \left<v_{i},x\right> = 0 \; \forall i \in I^{(t)}\}$. Here $I^{(t)} \subseteq [m]$ are the $|I^{(t)}| = \frac{n}{16}$ indices with maximum weight $w_i^{(t)}$.  
\item $U_4^{(t)} := \{ x \in \setR^n \mid \left<v_i,x\right> = 0 \;\; \forall i\textrm{ with }\lambda_i \leq 1\}$
\item $U_5^{(t)} := \{ x \in \setR^n \mid \big<x, \sum_{i=1}^m \lambda_iw_i^{(t)}\cdot \exp\left(-\frac{4\delta^2\lambda_i^2}{n}\right) v_i \big> = 0\}$
\item $U_6^{(t)} := \textrm{span}\{ u_j(M^{(t)}) : \frac{1}{16}n\leq j \leq n\}$, for $M^{(t)} := \sum_{i=1}^m w_i^{(t)} \lambda_i^2 v_iv_i^T$.
\item $U^{(t)} := U_1^{(t)} \cap \ldots \cap U_6^{(t)}$
\end{itemize*}
  \item[(4)] Let $z^{(t)}$ be any unit vector in $U^{(t)}$ 
  \item[(5)] Choose a maximal $\alpha^{(t)} \in (0,1]$ so that $x^{(t+1)}:=x^{(t)} + \delta \cdot y^{(t)} \in [-1,1]^n$, with $y^{(t)} = \alpha^{(t)}z^{(t)}$.
  \item[(6)] Update $w_i^{(t+1)} := w_i^{(t)} \cdot \exp(\lambda_i \cdot \delta \cdot \left<v_i, y^{(t)}\right>) \cdot \exp\left(-\frac{4\delta^2\lambda_i^2}{n}\right)$.
  \item[(7)] Let $A^{(t)} := \{ j \in [n] : -1 < x_j^{(t)} < 1 \}$. If $|A^{(t)}| < \frac{n}{2}$, then set $T := t$ and stop.
  \end{enumerate}
\end{enumerate}
The intuition is that we maintain weights $w_i^{(t)}$ for each constraint $i$ that increase exponentially
with the one-sided discrepancy $\left<v_i,x^{(t)}-x^{(0)}\right>$. Those weights are discounted in each iteration by a factor that is slightly less than 1 --- with a bigger discount for constraints with a larger parameter $\lambda_i$.
The subspaces $U_1^{(t)}$ and $U_2^{(t)}$ ensure that the length of $x^{(t)}$ is monotonically increasing and
fully colored elements remain fully colored.

\subsection{Bounding the number of iterations}

First, note that if the algorithm terminates, then 
at least half of the variables in $x^{(T)}$ will be either $-1$ or $+1$. In particular, once a
variable is set to $\pm 1$, it is removed from the set $A^{(t)}$ of active variables and the subsequent updates will leave those coordinates invariant. 

First we bound the number of iterations. Here we use that the algorithm always makes a 
step of length $\delta$ orthogonal to the current position --- except for the steps where
it hits the boundary. 
\begin{lemma} \label{lem:BoundOnNumOfIterations}
The algorithm terminates after $T = O(\frac{n}{\delta^2})$ iterations.
\end{lemma}
\begin{proof}
First, we can analyze the length increase
\[
 \|x^{(t+1)}\|_2^2 = \|x^{(t)} + \delta \cdot y^{(t)}\|_2^2 = \|x^{(t)}\|_2^2 + 2\delta \underbrace{\big<x^{(t)},y^{(t)}\big>}_{=0} + \delta^2 \|y^{(t)}\|_2^2,
\]
using that $y^{(t)} \in U_2^{(t)}$.
Whenever $\alpha^{(t)} = 1$, we have $\|x^{(t+1)}\|_2^2 \geq \|x^{(t)}\|_2^2 + \delta^2$.
It happens that $\alpha^{(t)} < 1$ at most $n$ times, simply because in each such iteration
$|A^{(t)}|$ must decrease by at least one. We know that $x^{(T)} \in [-1,1]^n$.
Suppose for the sake of contradiction that $T > \frac{2n}{\delta^2}$, then $\|x^{(T)}\|_2^2 \geq (T-n) \cdot \delta^2 > n$, which is impossible. We can hence conclude that the algorithm will terminate in step (7) after at most $\frac{2n}{\delta^2}$ iterations. 
\end{proof}

\subsection{Properties of the subspace \texorpdfstring{$U^{(t)}$}{U(t)}}

One obvious condition to make the algorithm work is to guarantee that the subspace
$U^{(t)}$ satisfies $\dim(U^{(t)}) \geq 1$. In fact, its dimension will even be linear in $n$.
\begin{lemma}
In any iteration $t$, one has  $\dim(U^{(t)}) \geq \frac{n}{8}$.
\end{lemma}
\begin{proof}
We simply need to account for all linear constraints that define $U^{(t)}$
and we get  
\[
  \dim(U^{(t)}) \geq |A^{(t)}| - |I^{(t)}| - |\{ i: \lambda_i \leq 1\}| -\frac{n}{16} - 2 \geq \frac{n}{2} - \frac{n}{16}-\frac{n}{8} - \frac{n}{16} - 2 \geq \frac{n}{8}
\]
assuming that $n \geq 16$.
\end{proof}
Another crucial property will be that every vector in $U^{(t)}$ has a bounded \emph{quadratic error term}:
\begin{lemma} \label{lem:QuadraticErrorInSubspaceU}
For each unit vector $y \in U^{(t)}$ one has $y^TM^{(t)}y \leq \frac{16}{n} \sum_{i=1}^m w_i^{(t)}\lambda_i^2$. 
\end{lemma}
\begin{proof}
We have $\textrm{Tr}[v_iv_i^T] = 1$ since each $v_i$ is a unit vector, hence $\textrm{Tr}[M^{(t)}] = \sum_{i=1}^m w_i^{(t)} \lambda_i^2 \textrm{Tr}[v_iv_i^T] = \sum_{i=1}^m w_i^{(t)} \lambda_i^2.$ Because $M^{(t)}$ is positive semidefinite, we know that $\mu_1,\ldots,\mu_n \geq 0$, where $\mu_j := \mu_j(M^{(t)})$ is the $j$th
eigenvalue. Then 
by \emph{Markov's inequality} at most a $\frac{1}{16}$ fraction of eigenvalues can be larger than $\frac{16}{n} \cdot \textrm{Tr}[M^{(t)}]$. The claim follows as $U_6^{(t)}$ is spanned by the $\frac{15}{16}n$
eigenvectors $v_j(M^{(t)})$ belonging to the smallest eigenvalues, which means $\mu_j \leq \frac{16}{n} \textrm{Tr}[M^{(t)}]$ for $j=\frac{1}{16}n,\ldots,n$.
\end{proof}

\subsection{The potential function}

So far, we have defined the weights by iterative update steps, but it is not
hard to verify that in each iteration $t$ one has the explicit expression
\begin{equation} \label{eq:ExplicitExpressionForWit}
  w_i^{(t)} = \exp\Big(\lambda_i\big<v_i,x^{(t)}-x^{(0)}\big> - \lambda_i^2 \cdot \Big(1 + t \cdot \frac{4\delta^2}{n}\Big)\Big).
\end{equation}
Inspired by the multiplicative weight update method, we consider the \emph{potential function} $\Phi^{(t)} := \sum_{i=1}^m w_i^{(t)}$ that is simply the sum of the individual weights. 
At the beginning of the algorithm we have $\Phi^{(0)} = \sum_{i=1}^m w_i^{(0)} = \sum_{i=1}^m \exp(-\lambda_i^2/16) \leq \frac{n}{32}$ using the assumption in Theorem~\ref{thm:DeterministicLovettMeka}. 
Next, we want to show that the potential function does not increase. Here the choice of the
subspaces $U_5^{(t)}$ and $U_6^{(t)}$ will be crucial to control the error.

\begin{lemma}\label{PotentialFunctionBounded} 
In each iteration $t$ one has $\Phi^{(t+1)} \leq \Phi^{(t)}$.
\end{lemma}
\begin{proof}
Let us abbreviate $\rho_i := \exp\left(-\frac{4\delta^2\lambda_i^2}{n}\right)$ as the \emph{discount factor}
for the $i$th constant. Note that in particular $0 < \rho_i \leq 1$ and $\rho_i \leq 1-\frac{2\delta^2 \lambda_i^2}{n}$. 
The change in one step can be analyzed as follows:
\begin{eqnarray*}
 \Phi^{(t+1)} 
 &=& \sum_{i=1}^m w_i^{(t+1)} = \sum_{i=1}^m w_i^{(t)} \cdot \exp\big(\lambda_i \delta \big<v_i,y^{(t)}\big>\big) \cdot \rho_i \\
&\stackrel{(*)}{\leq}& \sum_{i=1}^m w_i^{(t)} \cdot \Big( 1 + \lambda_i \delta \big<v_i,y^{(t)}\big> + \lambda_i^2 \delta^2 \big<v_i,y^{(t)}\big>^2 \Big) \cdot \rho_i \\
&=& \sum_{i=1}^m w_i^{(t)} \cdot \rho_i + \delta \underbrace{\Big<\sum_{i=1}^m \lambda_iw_i^{(t)}\rho_iv_i,y^{(t)}\Big>}_{=0\textrm{ since }y^{(t)} \in U_5^{(t)}} + \delta^2 \sum_{i=1}^m w_i^{(t)}\lambda_i^2\underbrace{\rho_i}_{\leq 1} \big<v_i,y^{(t)}\big>^2  \\
&\leq&  \sum_{i=1}^m w_i^{(t)} \cdot \rho_i + \delta^2 \cdot (y^{(t)})^TM^{(t)}y^{(t)} \stackrel{(**)}{\leq} \sum_{i=1}^m w_i^{(t)} \cdot \rho_i + \delta^2 \frac{16}{n}\sum_{i=1}^m w_i^{(t)} \lambda_i^2 \\
&\stackrel{(***)}{\leq}& \sum_{i=1}^m w_i^{(t)} = \Phi^{(t)}.
\end{eqnarray*}
In $(*)$, we use the inequality $e^x\leq 1+x+x^2$ for $|x| \leq 1$ together with the fact that $\lambda_i \delta |\left<v_i,y^{(t)}\right>| \leq \lambda_i \delta \leq 1$. In $(**)$ we bound $(y^{(t)})^T M^{(t)} y^{(t)}$ using Lemma~\ref{lem:QuadraticErrorInSubspaceU}.
In $(***)$ we finally use the fact that $\rho_i + \frac{16}{n}\delta^2 \leq 1$.
\end{proof}

Typically in the multiplicative weight update method one can only use the 
fact that $\max_{i\in [m]} w_i^{(t)} \leq \Phi^{(t)}$ which would lead to the loss
of an additional $\sqrt{\log n}$ factor. 
The trick in our approach is that there is always a \emph{linear} number of weights of order $\max_{i\in [m]} w_i^{(t)}$
since the updates are always chosen orthogonal to the $\frac{n}{16}$ constraints with highest weight.

\begin{lemma} At the end of the algorithm, $\max\{ w_i^{(T)} : i \in [m] \} \leq 2.$
\end{lemma}
\begin{proof}
Suppose, for contradiction, that $w_i^{(T)} > 2$ for some $i$. Let $t^*$ be the last iteration when $i$ was not among the $\frac{n}{16}$ constraints with highest weight. After iteration $t^*+1$, $w_i^{(t)}$ only decreases in each iteration, due to the factor $\exp\left(-\frac{4\delta^2\lambda_i^2}{n}\right)$. Then
\[
2 < w_i^{(T)} = w_i^{(t^* + 1)} = w_i^{(t^*)} \cdot \underbrace{\exp(\lambda_i \cdot \delta \cdot \big<v_i, y^{(t)}\big>)}_{\leq e} \cdot \underbrace{\rho_i}_{\leq 1} \leq  w_i^{(t^*)}\cdot e,
\]
and hence,
$
w_i^{(t^*)} > \frac{2}{e}.
$
This would imply that
$
\Phi^{(t^*)} \geq \frac{n}{16} \cdot \frac{2}{e} > \frac{n}{32},
$
contradicting Lemma~\ref{PotentialFunctionBounded}.
\end{proof}

\begin{lemma}
If $w_i^{(T)} \leq 2$, then $\left<v_i,x^{(T)}-x^{(0)}\right> \leq 11 \lambda_i$.
\end{lemma}
\begin{proof}
First note that the algorithm always walks orthogonal to all constraint vectors $v_i$ if $\lambda_i \leq 1$
and in this case $\left<v_i,x^{(T)} - x^{(0)}\right> = 0$. Now suppose that $\lambda_i > 1$. 
We know that $w_i^{(T)} \stackrel{\eqref{eq:ExplicitExpressionForWit}}{=} \exp\Big(\lambda_i\cdot \big<v_i,x^{(T)}-x^{(0)}\big> - \lambda_i^2 \cdot \Big(1 + 4 \cdot T \cdot \frac{\delta^2}{n}\Big) \Big) \leq 2.$
Taking logarithms on both sides and dividing by $\lambda_i$ then gives
\[
\big<v_i,x^{(T)} - x^{(0)}\big> \leq  \underbrace{\frac{\log(2)}{\lambda_i}}_{\leq 2} + \lambda_i \Big(1 + 4 \underbrace{T \frac{\delta^2}{n}}_{\leq 2}\Big) \leq 11\lambda_i.
\]
This lemma concludes the proof of Theorem \ref{thm:DeterministicLovettMeka}.
\end{proof}

\subsection{Application to set coloring}

Now we come to the main application of the partial coloring argument from Theorem~\ref{thm:DeterministicLovettMeka}, which is 
to color set systems:
\begin{lemma}\label{lem:SetSystemPartial}
Given a set system $S_1,\ldots,S_m \subseteq [n]$, we can find a coloring $x \in \{ -1,1\}^n$ with $|\sum_{j \in S_i} x_j| \leq O(\sqrt{n \log \frac{2m}{n}})$ for every $i$ deterministically in time $O\big(n^3m\log(\frac{2m}{n})\big)$. 
\end{lemma}
\begin{proof}
For a fractional vector $x$, let us abbreviate $\textrm{disc}(S,x) := |\sum_{j \in S} x_j|$ as the discrepancy with respect to set $S$.
Set $x^{(0)} := \bm{0}$. For $s=1,\ldots,\log_2(n)$ many phases we do the following. Let $A^{(s)} := \{ i \in [n] : -1 < x^{(s-1)}_i < 1\}$
be the not yet fully colored elements. Define a vector $v_i := \frac{1}{\sqrt{|A^{(s)}|}} \bm{1}_{S_i \cap A^{(s)}}$
of length $\|v_i\|_2 \leq 1$ with
parameters $\lambda_i := C\sqrt{\log(\frac{2m}{|A^{(s)}|})}$. Then apply Theorem~\ref{thm:DeterministicLovettMeka} to 
find $x^{(s)} \in [-1,1]^n$ with $\textrm{disc}(S_i,x^{(s)}-x^{(s-1)}) \leq O(\sqrt{|A^{(s)}| \log(\frac{2m}{|A^{(s)}|})})$ such that $x^{(s)}_i=x^{(s-1)}_i$ for $i\not\in A^{(s)}$. Since each time at least half of the elements get fully colored we have $|A^{(s)}| \leq 2^{-(s-1)}n$ for all $s$.
Then 
$x := x^{(\log_2 n)} \in \{ -1,1\}^n$ and 
\[
 \textrm{disc}(S_i,x) \leq \sum_{s \geq 1} O\Big(\sqrt{2^{-(s-1)}n \log\Big(\frac{2m}{2^{-(s-1)}n)}\Big)}\Big) \leq O\Big(\sqrt{n \log(\tfrac{2m}{n}})\Big)
\]
using that this convergent sequence is dominated by the first term.

In each application of Theorem~\ref{thm:DeterministicLovettMeka} one has $\delta \geq \Omega(1 / \sqrt{\log(\frac{2m}{n})})$. Thus phase $s$ runs for $O(2^{-(s-1)}n/\delta^2)=O(2^{-(s-1)}n\log(\frac{2m}{n}))$ iterations, each of which takes $O((2^{-(s-1)}n)^2m)$ time. This gives a total runtime of $O((2^{-(s-1)}n)^3m\log(\frac{2m}{n}))$ in phase $s$. Summing the geometric series for $s=1,\ldots,\log_2 n$ results in a total running time of $O(n^3m \log(\frac{2m}{n}))$.
\end{proof}
By setting $m=n$ in Lemma \ref{lem:SetSystemPartial}, we obtain Corollary \ref{lem:RunningTimeSpencersTheorem}.

\section{Matrix balancing}

In this section we prove Theorem~\ref{thm:MatrixBalancing}. We begin with some preliminaries. For matrices $A,B \in \setR^{n \times n}$, 
let $A \bullet B := \sum_{i=1}^n \sum_{j=1}^n A_{ij} \cdot B_{ij}$ be the \emph{Frobenius inner product}.
Recall that any symmetric matrix $A \in \setR^{n \times n}$ 
can be written as $A = \sum_{j=1}^n \mu_ju_ju_j^T$, where $\mu_j$ is the eigenvalue corresponding to eigenvector $u_j$. The \emph{trace} of $A$ is 
$\textrm{Tr}[A] = \sum_{i=1}^n A_{ii} = \sum_{j=1}^n \mu_j$ and for symmetric matrices $A,B$ one has $\textrm{Tr}[AB] = A \bullet B$.
If $A$ has only nonnegative eigenvalues, we say that $A$ is \emph{positive semidefinite} and write $A \succeq 0$. Recall that $A \succeq 0$ if and only
if $y^TAy \geq 0$ for all $y \in \setR^n$. 
For a symmetric matrix $A$, we denote $\mu_{\max} := \max\{ \mu_j : j=1,\ldots,n\}$
as the largest Eigenvalue and $\|A\|_{\textrm{op}} := \max\{ |\mu_j| : j=1,\ldots,n\}$ as the largest singular value. Note that if $A \succeq 0$, then $|A\bullet B| \leq \textrm{Tr}[A] \cdot \|B\|_{\textrm{op}}$. If $A,B \succeq 0$, then $A \bullet B \geq 0$. Finally, note that for any symmetric matrix $A$ one has $A^2 := AA \succeq 0$.

From the eigendecomposition $A = \sum_{j=1}^n \mu_ju_ju_j^T$, one can easily show that the maximum singular value also satisfies
$\|A\|_{\textrm{op}} = \max\{ \|Ay\|_2 : \|y\|_2 = 1\}$ and $\|A\|_{\textrm{op}} = \max\{ |y^TAy| : \|y\|_2 = 1\}$.
For any function $f : \setR \to \setR$ we define $f(A) := \sum_{j=1}^n f(\mu_j) u_ju_j^T$
to be the symmetric matrix that is obtained by applying $f$ to all Eigenvalues. In particular 
we will be interested in the \emph{matrix exponential} $\exp(A) := \sum_{j=1}^n e^{\mu_j} u_ju_j^T$.
For any symmetric matrices $A,B \in \setR^n$, the \emph{Golden-Thompson inequality} says that
$\textrm{Tr}[\exp(A + B)] \leq \textrm{Tr}[\exp(A)\exp(B)]$. (It is not hard to see that for diagonal matrices one has equality.) We refer to the textbook of Bhatia~\cite{MatrixAnalysisBhatia1997}
for more details.
\begin{theorem}\label{thm:MatrixBalancingPartial}
Let $A_1,\ldots,A_n \in \setR^{m \times m}$ be $q$-block diagonal matrices with $\|A_i\|_{\textrm{op}} \leq 1$
for $i=1,\ldots,m$ and let $x^{(0)} \in [-1,1]^n$ be a starting point. Then there is a deterministic
algorithm that finds an $x \in [-1,1]^n$ with 
\[
  \Big\|\sum_{i=1}^n (x_i-x^{(0)}_i) \cdot A_i \Big\|_{\textrm{op}} \leq O\Big(\sqrt{n \log\Big(\frac{2qm}{n}\Big)}\Big)
\]
in time $O(n^5 + n^4m^3)$. Moreover, at least $\frac{n}{2}$ coordinates of $x$ will be in $\{ -1,1\}$.
\end{theorem}
Our algorithm computes a sequence
of iterates $x^{(0)},\ldots,x^{(T)}$ such that $x^{(T)}$ is the desired vector $x$ with half of the coordinates being integral. 
In our algorithm the step size is $\delta=\frac{1}{\sqrt{n}}$ and we use a parameter $\varepsilon=\frac{1}{\sqrt{n}}$ to control the scaling of the following potential function:
\[
 \Phi^{(t)} := \textrm{Tr}\Big[\exp\Big(\varepsilon \sum_{i=1}^n (x^{(t)}_i-x^{(0)}_i) \cdot A_i\Big)\Big].
\]
Suppose $B_{i,k} \in \setR^{q \times q}$ are symmetric matrices so that $A_i = \textrm{diag}(B_{i,1},\ldots,B_{i,m/q})$. Then we can decompose the weight function as $ \Phi^{(t)} = \sum_{k=1}^{m/q} \Phi^{(t)}_k$ with $\quad \Phi^{(t)}_k := \textrm{Tr}\Big[\exp\Big(\varepsilon \sum_{i=1}^n (x_i^{(t)} - x_i^{(0)})  B_{i,k}\Big)\Big].$
In other words, the potential function is simply the sum of the potential function applied to each individual block.
The algorithm is as follows:
\begin{enumerate}
\item[(1)] FOR $t = 0$ TO $\infty$ DO
  \begin{enumerate}
  \item[(2)] Define weight matrix $W^{(t)} := \exp(\varepsilon \sum_{i=1}^n (x^{(t)}_i - x^{(0)}_i) A_i)$ 
  \item[(3)] Define the following subspaces
\begin{itemize*}
\item $U_1^{(t)} := \textrm{span}\{ e_j : -1 < x^{(t)}_j < 1\}$
\item $U_2^{(t)} := \{ x \in \setR^n\mid \big<x,x^{(t)}\big> = 0\}$
\item $U_3^{(t)} := \{ x \in \setR^n \mid \sum_{i=1}^n x_iB_{i,k} = \bm{0} \; \forall k \in I^{(t)}\}$. Here $I^{(t)} \subseteq [m]$ are the $|I^{(t)}| = \frac{1}{16} \cdot \frac{n}{q^2}$ indices $k$ with maximum weight $\Phi_k^{(t)}$.
\item $U_4^{(t)} := \{ x \in \setR^n \mid \sum_{i=1}^n x_i \cdot (W^{(t)} \bullet A_i) = 0\}$
\item $U_5^{(t)}$ is the subspace defined in Lemma~\ref{lem:SubspaceWithBoundedMatrixQuadraticError}, with $k=16$.
\item $U^{(t)} := U_1^{(t)} \cap \ldots \cap U_5^{(t)}$
\end{itemize*}
  \item[(4)] Let $z^{(t)}$ be any unit vector in $U^{(t)}$.
  \item[(5)] Choose a maximal $\alpha^{(t)} \in (0,1]$ so that $x^{(t+1)}:= x^{(t)} + \delta \cdot y^{(t)} \in [-1,1]^n$, where $y^{(t)} = \alpha^{(t)}z^{(t)}$.
  \item[(6)] Let $A^{(t)} := \{ j \in [n] : -1 < x_j^{(t)} < 1 \}$. If $|A^{(t)}| < \frac{n}{2}$, then set $T := t$ and stop.
  \end{enumerate}
\end{enumerate}
The analysis of our algorithm follows a sequence of lemmas, the proofs of most of which we defer to Appendix \ref{MatrixProofs}. By exactly the same arguments as in Lemma ~\ref{lem:BoundOnNumOfIterations} we know that the algorithm terminates after $T \leq \frac{2n}{\delta^2}$
iterations. Each iteration can be done in time $O(n^2m^3 + n^3)$ (c.f. Lemma~\ref{lem:SubspaceWithBoundedMatrixQuadraticError}).
\begin{lemma}
In each iteration $t$ one has $\dim(U^{(t)}) \geq \frac{n}{4}$.
\end{lemma}
\begin{proof}
  We simply need to account for all linear constraints that define $U^{(t)}$
  and we get  
  \[
    \dim(U^{(t)}) \geq \underbrace{|A^{(t)}|}_{U_1^{(t)}} - \underbrace{|I^{(t)}|}_{U_3^{(t)}} -\underbrace{\frac{n}{16}}_{U_5^{(t)}} - \underbrace{2}_{U_2^{(t)},U_4^{(t)}} \geq \frac{n}{2} - \frac{n}{16q^2}\cdot q^2 - \frac{n}{16} - 2 \geq \frac{n}{4}
  \]
  assuming that $n \geq 16$.
\end{proof}

To analyze the behavior of the potential function, we first prove the existence of a
suitable subspace $U_5^{(t)}$ that will bound the quadratic error term. 
\begin{lemma} \label{lem:SubspaceWithBoundedMatrixQuadraticError}
Let $W \in \setR^{m \times m}$ be a symmetric positive semidefinite matrix, let $A_1,\ldots,A_n \in \setR^{m \times m}$
be symmetric matrices with $\|A_i\|_{\textrm{op}} \leq 1$ and let $k>0$ be a parameter. Then in time $O(n^2m^3+n^3)$ one can compute a subspace $U \subseteq \setR^n$
of dimension $\dim(U) \geq (1-\frac{1}{k})n$ so that 
\begin{equation} \label{eq:MatrixQuadraticErrorIneq}
  W \bullet \Big(\sum_{i=1}^n y_iA_i\Big)^2 \leq k\cdot \textrm{Tr}[W] \quad \forall y \in U\textrm{ with }\|y\|_2 = 1.
\end{equation}
\end{lemma}
{\em Proof.} See Appendix \ref{MatrixProofs}.
\vspace{2mm}

Again, we bound the increase in the potential function: 
\begin{lemma}\label{lem:PotentialFunctionDecreaseMatrixBalancing}
In each iteration $t$, one has $\Phi^{(t+1)} \leq (1 + 16\varepsilon^2\delta^2) \cdot \Phi^{(t)}$.
\end{lemma}
{\em Proof.} See the Appendix \ref{MatrixProofs}.
\vspace{2mm}

This gives us a bound on the potential function at the end of the algorithm.
\begin{lemma}
At the end of the algorithm, $\Phi^{(T)} \leq m \cdot \exp(32\varepsilon^2 n)$.
\end{lemma}
\begin{proof}
Since $\Phi^{(0)} = \textrm{Tr}[\exp(\bm{0})] = \textrm{Tr}[I] = m$, we get that 
$\Phi^{(T)} \leq m \cdot (1+16\varepsilon^2\delta^2)^T \leq m \cdot \exp(32\varepsilon^2 n)$, using the fact that $T \leq \frac{2n}{\delta^2}$.
\end{proof}

\begin{lemma}\label{lem:muMax}
We have $\mu_{\max}(\sum_{i=1}^n (x_i^{(T)} - x_i^{(0)}) \cdot A_i)=O(\sqrt{n \log(\frac{2qm}{n})})$.
\end{lemma}
{\em Proof.} See Appendix \ref{MatrixProofs}.
\vspace{2mm}

These lemmas put together give us Theorem~\ref{thm:MatrixBalancingPartial}: an algorithm that yields a partial coloring with the claimed properties. We run the algorithm in phases to obtain Theorem \ref{thm:MatrixBalancing}, by boosting the partial coloring to a full coloring using a similar technique as in Lemma~\ref{lem:SetSystemPartial}. The interested reader may refer to Appendix \ref{MatrixProofs} for details.

\bibliographystyle{alpha}
\bibliography{main}

\newpage

\appendix

\section{Proofs from Section 3}\label{MatrixProofs}

  \begin{proof}[of Lemma~\ref{lem:SubspaceWithBoundedMatrixQuadraticError}]
Let $M \in \setR^{n \times n}$ be the matrix with entries $M_{ij}:= W\bullet A_iA_j$ for all $i,j \in [n]$. 
First note that the matrix $M$ is symmetric\footnote{One has to be careful as the product $A_iA_j$ is in general \emph{not} symmetric, even if $A_i$ and $A_j$ are symmetric.}, because
    $
      M_{ij}=W\bullet A_iA_j=W^T\bullet (A_iA_j)^T=W\bullet A_jA_i=M_{ji}$. 
Now, for any $y \in \setR^n$,  $\left(\sum_{i=1}^n y_iA_i\right)^2$ is symmetric and positive semidefinite and hence $y^TMy=\sum_{i,j=1}^ny_iy_j (W\bullet A_iA_j) =W\bullet \Big(\sum_{i=1}^ny_iA_i\Big)^2 \geq 0$, proving that $M$ is positive semidefinite. 

  Consider the eigendecomposition $M=\sum_{i=1}^n \mu_i\; u_iu_i^T$ where $\mu_i\geq 0$. Define the subspace
    $
      U:=\textrm{span}\{u_i \colon \mu_i < k \mathrm{Tr}[W]\}.
    $
    The desired inequality \eqref{eq:MatrixQuadraticErrorIneq} follows immediately from the definition of $U$. All that remains is to verify that $\dim(U)\geq (1-\frac{1}{k})n$.
    Since $\mu_i\geq 0$ we may apply {\em Markov's inequality} to deduce that
    $
      \# \{i\colon \mu_i\geq k\cdot\textrm{Tr}[W] \}\leq \frac{\textrm{Tr}[M]}{k\textrm{Tr}[W]}\leq \frac{n}{k},
    $
    where in the second inequality we have used the bound
    $      \textrm{Tr}[M] = \sum_{i=1}^n M_{ii} = \sum_{i=1}^n W\bullet A_i^2\leq n \cdot \textrm{Tr}[W].
    $
    Hence $\dim(U)\geq n-\#\{i\colon \mu_i\geq k\textrm{Tr}[W]\}\geq (1-\frac{1}{k})n$, as desired. 
    
    Finally, to bound the running time, we observe that computing $M$ takes time $O(n^2m^3)$ and the eigendecomposition of $M$ can be computed in time $O(n^3)$.
  \end{proof}

\begin{proof}[Proof of Lemma~\ref{lem:PotentialFunctionDecreaseMatrixBalancing}]
We estimate that
\begin{eqnarray*}
 \Phi^{(t+1)} &=& \textrm{Tr}\Big[\exp\Big(\varepsilon \sum_{i=1}^n (x^{(t+1)}_i - x^{(0)}_i) A_i \Big)\Big] \\
&=& \textrm{Tr}\Big[\exp\Big(\varepsilon \sum_{i=1}^n (x^{(t)}_i - x^{(0)}_i) A_i + \varepsilon\delta \sum_{i=1}^n y_i^{(t)}\Big)\Big] \\
&\stackrel{(*)}{\leq}& \textrm{Tr}\Big[\underbrace{\exp\Big(\varepsilon \sum_{i=1}^n (x_i^{(t)} - x_i^{(0)}) \cdot A_i\Big)}_{=W^{(t)}} \exp\Big(\varepsilon\delta \sum_{i=1}^n y_i^{(t)}A_i\Big)\Big] \\
&=& W^{(t)} \bullet \exp\Big( \varepsilon\delta \sum_{i=1}^n y_i^{(t)} A_i\Big) \\
&\stackrel{(**)}{\leq}& W^{(t)} \bullet \Big(I +  \varepsilon\delta \sum_{i=1}^n y_i^{(t)} A_i + \varepsilon^2\delta^2 \Big(\sum_{i=1}^n y_i^{(t)} A_i\Big)^2 \Big) \\
&=& \underbrace{W^{(t)} \bullet I}_{=\Phi^{(t)}} + \varepsilon\delta \underbrace{\Big(W^{(t)} \bullet \sum_{i=1}^n y_i^{(t)}A_i\Big)}_{=0\textrm{ since }y^{(t)} \in U_4^{(t)}} + \varepsilon^2\delta^2 \underbrace{\Big( W^{(t)} \bullet \Big(\sum_{i=1}^n y_i^{(t)} A_i\Big)^2 \Big)}_{\leq 16\cdot \textrm{Tr}[W^{(t)}]} \\
&\stackrel{(***)}{\leq}& \Phi^{(t)} \cdot (1 + 16\varepsilon^2\delta^2).
\end{eqnarray*}
In $(*)$ we use the Golden-Thompson inequality. In $(**)$ we use that $\exp(X) \preceq I + X + X^2$ for any symmetric 
matrix $X$ with $\|X\|_{\textrm{op}} \leq 1$ together with the triangle inequality
$$
\left\|\varepsilon\delta\sum_{i=1}^n y_i^{(t)} A_i\right\|_{op}\leq \varepsilon\delta\sum_{i=1}^n|y_i^{(t)}|\cdot\|A_i\|_{op}\leq \varepsilon\delta n=1.
$$
In $(***)$ we use Lemma~\ref{lem:SubspaceWithBoundedMatrixQuadraticError} and the fact that $y^{(t)} \in U_5^{(t)}$.
\end{proof}

  \begin{proof}[Proof of Lemma~\ref{lem:muMax}]
    Let $\mu_{\max}:=\mu_{\max}(\sum_{i=1}^n (x_i^{(T)} - x_i^{(0)}) \cdot A_i)$. Suppose the eigenspace corresponding to $\mu_{\textrm{max}}$ lies in block $k$, for some $k \in \lbrace{1,\ldots,m/q\rbrace}$. Let $t^*$ be the last iteration when $k$ was not among the $n/(16q^2)$ indices with maximum weight. We then have
\begin{align*}
\Phi^{(T)}_k &= \Phi^{(t^*+1)}_k = \textrm{Tr}\Big[\exp\Big(\varepsilon \sum_{i=1}^n (x_i^{(t^*)} + \delta y_i^{(t^*)} - x_i^{(0)})  B_{i,k}\Big)\Big] \\ 
&\stackrel{(*)}{\leq} \textrm{Tr}\Big[\exp\Big(\varepsilon \sum_{i=1}^n (x_i^{(t^*)} - x_i^{(0)})  B_{i,k}\Big)\underbrace{\exp\Big(-\varepsilon\delta \sum_{i=1}^n y_i^{(t^*)}B_{i,k}\Big)}_{\preceq e\bf I\quad (**)} \Big] \leq e\cdot\Phi^{(t^*)}_k,
\end{align*}
where in $(*)$, we use the Golden-Thompson inequality. In $(**)$ we have used the bounds $\|B_{i,k}\|_{op}\leq \|A_i\|_{op}\leq 1$ and $|y_i^{(t^*)}|\leq 1$ together with the triangle inequality to deduce that $\left\|\varepsilon\delta\sum_{i=1}^ny_i^{(t^*)}B_{i,k}\right\|_{op}\leq 1$. Hence
\[
  e^{\varepsilon \mu_{\max}} \leq \Phi^{(T)}_k \leq e\cdot \Phi^{(t^*)}_k \leq e\cdot \frac{16q^2}{n} \cdot \Phi^{(T)} \leq \frac{16eq^2}{n} \cdot m\exp(32\varepsilon^2 n).
\]
Then taking logarithms and dividing by $\varepsilon$ gives
\[
  \mu_{\max} \leq \frac{1}{\varepsilon} \cdot \log\left( \frac{16eq^2m}{n} \right) + 32\varepsilon n = O\Big(\sqrt{n\log(\frac{qm}{n})}\Big),
\] 
where in the final inequality we have used that $\varepsilon=\sqrt{\frac{\log(qm/n)}{n}}$.
\end{proof}


\begin{proof}[Proof of Theorem~\ref{thm:MatrixBalancing}]
  Set $x^{(0)} := \bm{0}$. For $s=1,\ldots,\log_2(n)$ many phases we do the following. Let $J^{(s)} := \{ i \in [n] : -1 < x^{(s-1)}_i < 1\}$
  be the not yet fully colored elements. Apply Theorem~\ref{thm:MatrixBalancingPartial} to find $x^{(s)} \in [-1,1]^n$ with 
  $$
    \left\|\sum_{i\in J^{(s)}} (x_i^{(s)}-x_i^{(s-1)} )\cdot A_i\right\|_{op}=O\left(\sqrt{|J^{(s)}| \log\frac{2qm}{|J^{(s)}|}}\right),
  $$
  and such that $x^{(s)}_i=x^{(s-1)}_i$ for all $i\not\in J^{(s)}$. Since each time at least half of the elements get fully colored we have $|J^{(s)}| \leq 2^{-(s-1)}n$ for all $s$.
  Then 
  $x := x^{(\log_2 n)} \in \{ -1,1\}^n$ and 
  \[
   \left\|\sum_{i=1}^n x_i\cdot A_i\right\|_{op} = \sum_{s \geq 1} O\Big(\sqrt{2^{-(s-1)}n \log\Big(\frac{2qm}{2^{-(s-1)}n)}\Big)}\Big) = O\Big(\sqrt{n \log(\tfrac{2qm}{n}})\Big),
  \]
  using that the sum of a subgeometric sequence is dominated by its first term. Phase $s$ has a running time of $O((2^{-(s-1)}n)^5+(2^{-(s-1)}n)^4m^3)$ and summing this geometric series over $s=1,\ldots,\log_2 n$ yields a total runtime of $O(n^5+n^4m^3)$.
\end{proof}

\section{Discrepancy minimization for matrices with bounded column length}\label{BoundedColumns}

In this section we prove Theorem~\ref{thm:ConstructiveBeckFialaAlgorithm}.
Fix a matrix $A \in \setR^{m \times n}$ with $\|A^j\|_2 \leq 1$ for each column $j=1,\ldots,n$.
Recently Bansal, Dadush and Garg \cite{ConstructiveBanaBansalDadushGarg16} gave the first 
polynomial time algorithm to find a coloring $x \in \{ -1,1\}^n$ with $\|Ax\|_{\infty} \leq O(\sqrt{\log n})$. Their method is based on a random walk, where the random updates in each iteration are 
chosen using a semidefinite program that has to be re-solved each time. We show that instead 
a deterministic walk can be used, guided by a suitable exponential potential function. The update directions 
will be chosen from the intersection of subspaces satisfying certain constraints; no SDP has to be solved in our method. We should also mention that the more general non-constructive result of Banaszczyk~\cite{BalancingVectors-Banaszczyk98} even guarantees signs $x$ so that $Ax \in 5 \cdot K$, where $K$ is any convex
body with Gaussian measure at least $1/2$.

In this section, let $C>0$ be a sufficiently large constant. 
For a row $i$ with $\|A_i\|_2^2 \leq \frac{1}{n}$, any coloring $x$ will satisfy $|\left<A_i,x\right>| \leq \|A_i\|_2 \cdot \|x\|_2 \leq 1$ and we can safely remove such a row. From now on we can assume that $\|A_i\|_2^2 \geq \frac{1}{n}$
and hence $m \leq n^2$.  
Note that it also suffices to find an $x \in \{ -1,1\}^n$ satisfying the \emph{one-sided} error $\left<A_i,x\right> \leq O(\sqrt{\log n})$ as one can simply stack $\frac{1}{\sqrt{2}} A$ and $-\frac{1}{\sqrt{2}}A$ together. Next, replace each row
 $A_i$ with two rows: one row is the \emph{light} row containing all entries of size $|A_{ij}| \leq \frac{1}{C^2\sqrt{\log n}}$ and the other row is the \emph{heavy} row whose only nonzero entries have size $|A_{ij}| > \frac{1}{C^2\sqrt{\log n}}$.
After this modification, we abbreviate the indices as
 $I_{\textrm{light}} := \{ i \in [m] : \|A_i\|_{\infty} \leq \frac{1}{C^2\sqrt{\log n}} \}$ and $I_{\textrm{heavy}} := \{ i \in [m] : \|A_i\|_{\infty} > \frac{1}{C^2\sqrt{\log n}}\}$.
As in the previous settings, our algorithm will compute a sequence $x^{(0)},\ldots,x^{(T)} \in [-1,1]^n$, starting at $x^{(0)} = \bm{0}$ 
so that the final point $x^{(T)}$ has coordinates only in $\{ -1,1\}$. 
For the point $x^{(t)} \in [-1,1]^n$ and some parameters $\alpha,\beta>0$ that we specify later, 
we define a \emph{potential function} $\Phi^{(t)} := \sum_{i \in I_{\textrm{light}}} w_i^{(t)}$ with
\[w_i^{(t)} := \exp\Big(\alpha \left<A_i,x^{(t)}\right> + \beta \min\Big\{ C,\sum_{j=1}^n (1-(x_j^{(t)})^2) \cdot A_{ij}^2 \Big\} \Big).\]
Here the quantity $L(i,x) := \sum_{j=1}^n (1-x_j^2) \cdot A_{ij}^2$ can be interpreted as the \emph{effective length}
of row $i$ with $L(i,\bm{0}) = \|A_i\|_2^2$ and $L(i,x) = 0$, if $x \in \{ -1,1\}^n$. 

The intuition behind the algorithm is as follows: at the beginning one has 
 $x^{(0)} = \bm{0}$ and the whole 
weight of the potential function comes from the $\beta$-term. Then in the course of the
algorithm the weight is transferred from the $\beta$-term to the $\alpha$-term until all elements
are colored and the effective length of all constraints is $0$. In fact, if $\beta \geq \Omega(\alpha^2)$, we show that the potential function is nonincreasing.

To keep the notation readable, for vectors $x,y \in \setR^n$ we write 
$(x \circ y) \in \setR^n$ for the vector with components $(x \circ y)_i := x_i \cdot y_i$ and $x^{\circ 2} := x \circ x$. Moreover, $x^{\otimes 2} := xx^T$ is the
\emph{tensor product}.
As before, we find an update vector in each iteration so that the potential function does not
 increase, by choosing it from the intersection of certain subspaces.
We postpone some linear algebra arguments till Section~\ref{sec:BeckFialaQuadraticErrorInSubspace}. 
We use the following algorithm: 
\begin{enumerate*}
\item[(1)] Set $x^{(0)} := \bm{0}$ and $A^{(0)} := [n]$.
\item[(2)] FOR $t=0$ TO $T$ DO
  \item[(3)] Let $I^{(t)} := \{ i \in I_{\textrm{light}} \mid L(i,x^{(t)}) < C\} \cup \{ i \in I_{\textrm{heavy}} \mid \sum_{j\in A^{(t)}} A_{ij}^2 < C\}$. Define the subspaces
  \begin{itemize*}
  \item $U_0^{(t)} := \textrm{span}\{ e_j \mid j \in A^{(t)} \} \subseteq \setR^n$
  \item $U_1^{(t)} := \{ x \in U_0^{(t)} \mid \left<x,x^{(t)}\right> = 0\}$
  \item $U_2^{(t)} := \{ x \in U_{0}^{(t)}  \mid \left<x, A_i\right> = 0 \; \forall i \notin I^{(t)}\}$
  \item $U_3^{(3)} := \{ x \in U_0^{(t)} \mid \left<\sum_{i \in I^{(t)}} w_i^{(t)}A_i,x\right> = 0 \}$
  \item $U_4^{(t)} := \{ x \in U_0^{(t)} \mid  \left<\sum_{i \in I^{(t)}} w_i^{(t)}\cdot (A_i^{\circ 2} \circ x^{(t)}),x\right> = 0 \}$
  \item $U_5^{(t)} \subseteq U_0^{(t)}$ with $\sum_{i \in I^{(t)}} w_i^{(t)} \left<A_i,x\right>^2 \leq \frac{\beta}{16\alpha^2} \sum_{i \in I^{(t)}} w_i^{(t)} \sum_{j=1}^n x_j^2A_{ij}^2$ for all $x \in U_5^{(t)}$ and $\dim(U_5^{(t)}) \geq \frac{15}{16}|A^{(t)}|$. (see Sec.~\ref{sec:BeckFialaQuadraticErrorInSubspace})
  \item $U_6^{(t)} \subseteq U_0^{(t)}$ with $\dim(U_6^{(t)}) \geq \frac{15}{16}|A^{(t)}|$ and $\sum_{i \in I^{(t)}} w_i^{(t)} \left<A_i^{\circ 2} \circ x^{(t)},x\right>^2 \leq \frac{1}{8\beta} \sum_{i \in I^{(t)}} w_i^{(t)} \sum_{j=1}^n x_j^2A_{ij}^2$ for all $x \in U_6^{(t)}$ (see Sec.~\ref{sec:BeckFialaQuadraticErrorInSubspace})
  \item $U^{(t)} := U_1^{(t)} \cap \ldots \cap U_6^{(t)}$.
  \end{itemize*}
  \item[(4)] Let $z^{(t)}$ be any unit vector in $U^{(t)}$. 
  \item[(5)] Choose a maximal $\alpha^{(t)} \in (0,1]$ so that $x^{(t+1)} := x^{(t)} + \delta \cdot y^{(t)} \in [-1,1]^n$ with $y^{(t)} = \alpha^{(t)}z^{(t)}$.
  \item[(6)] Let $A^{(t)} := \{ j \in [n] : -1 < x_j^{(t)} < 1 \}$. If $|A^{(t)}| \leq C$, then set $T := t$ and stop.
\end{enumerate*}
Technically speaking, the final point $x^{(T)}$ still has a constant number of entries not in $\{ -1,1\}$ --- 
these entries can be rounded arbitrarily. 
The first step is to guarantee that the subspace $U^{(t)}$ is indeed non-empty in each iteration: 
\begin{lemma}
In each iteration $t$, we have $\dim(U^{(t)}) \geq \frac{1}{2} |A^{(t)}|$, if $C$ is chosen large enough.
\end{lemma}
\begin{proof}
Observe that for $i \in I_{\textrm{light}} \setminus I^{(t)}$, \[
\sum_{j\in A^{(t)}} A_{ij}^2 \geq \sum_{j\in A^{(t)}} (1-\big(x^{(t)})^2 \big)A_{ij}^2 = \sum_{j=1}^n (1-\big(x^{(t)})^2 \big)A_{ij}^2 = L(i,x^{(t)}) \geq C,\]
and hence $\sum_{j\in A^{(t)}} A_{ij}^2 \geq C$ holds for all $i \notin I^{(t)}$. Now, since the $l^2$-norm of each column $A^j$ is at most $1$, we have 
\[
 |A^{(t)}| \geq \sum_{j\in A^{(t)}} \underbrace{\sum_{i \in [m]}  A_{ij}^2}_{\leq 1} = \sum_{i \in [m]} \sum_{j\in A^{(t)}} A_{ij}^2 \geq \sum_{i \notin I^{(t)}} \underbrace{\sum_{j\in A^{(t)}} A_{ij}^2}_{\geq C} \geq C (m-|I^{(t)}|).
\]
Hence, $\textrm{codim}(U_2^{(t)}) \leq |A^{(t)}|/C$. We can hence bound
\[
\dim(U^{(t)}) \geq \underbrace{|A^{(t)}|}_{U_0^{(t)}} - \underbrace{\frac{|A^{(t)}|}{C}}_{U_2^{(t)}} - \underbrace{\frac{|A^{(t)}|}{16}}_{U_5^{(t)}} - \underbrace{\frac{|A^{(t)}|}{16}}_{U_6^{(t)}} - \underbrace{(1+1+1)}_{U_1^{(t)},U_3^{(t)},U_4^{(t)}} \geq \frac{|A^{(t)}|}{2},
\]
if $C$ is chosen large enough.
\end{proof}
As before, one always has $\|x^{(t+1)}\|_2^2 \geq \|x^{(t)}\|_{2}^2$ and in each but at most $n$ iterations 
one has $\|x^{(t+1)}\|_2^2 = \|x^{(t)}\|_2^2 + \delta^2$. Then the algorithm terminates after $T \leq n+\frac{n}{\delta^2} \leq \frac{2n}{\delta^2}$
iterations, given that $0 < \delta \leq 1$.  

The main part of the analysis lies in guaranteeing that the potential function is nonincreasing.
\begin{lemma}
Suppose that $\beta \geq C \cdot \alpha^2$ where $C>0$ is a large enough constant with $0<\delta \leq \frac{1}{36\sqrt{\beta}}$
and $\|A_i\|_{\infty} \leq \frac{1}{C\sqrt{\beta}}$ for $i \in I_{\textrm{light}}$.
Then in each iteration $t$ we have $\Phi^{(t+1)} \leq \Phi^{(t)}$. 
\end{lemma}
\begin{proof}
Note that $w_i^{(t+1)} \leq w_i^{(t)}$ for any light index with $L(i,x^{(t)}) \geq C$. In fact, one can 
only have strict inequality if $L(i,x^{(t)}) > C \geq L(i,x^{(t+1)})$.
Hence we only need to prove that $\sum_{i \in I^{(t)} \cap I_{\textrm{light}}} w_i^{(t+1)} \leq \sum_{i \in I^{(t)} \cap I_{\textrm{light}}} w_i^{(t)}$.
For ease of notation, we drop the index $t$ and also write 
$x' = x + \delta y$ instead of $x^{(t+1)} = x^{(t)} + \delta y^{(t)}$, and $I$ instead of $I^{(t)} \cap I_{\textrm{light}}$. We estimate that
\begin{eqnarray}
\sum_{i \in I} w_{i}^{(t+1)}  
&=& \sum_{i \in I} \exp\Big(\alpha \left<A_i,x+\delta y\right> + \beta \sum_{j=1}^n (1-(x_j+\delta y_j)^2) \cdot A_{ij}^2 \Big) \label{eq:BeckFialaPotFctBound} \\
&=& \sum_{i \in I} \underbrace{\exp\Big(\alpha \left<A_i,x\right> + \beta \sum_{j=1}^n (1-x_j^2) \cdot A_{ij}^2\Big)}_{=w_i} \\
&\cdot& \exp\Big(\alpha \delta \left<A_i,y\right> - 2\beta \delta \big<A_i^{\circ 2} \circ x,y\big>  - \beta \delta^2 \sum_{j=1}^ny_j^2A_{ij}^2\Big)  \nonumber
\end{eqnarray}
Now we bound the second exponential term using the inequality $e^{x_1+x_2+x_3} \leq 1+x_1+x_2+x_3 + 9x_1^2+9x_2^2+9x_3^2$ for $\max\{ |x_1|,|x_2|,|x_3|\} \leq 1$. We obtain

\begin{eqnarray*}
\eqref{eq:BeckFialaPotFctBound} &\leq& 
 \sum_{i \in I} w_i + \Big[ \alpha \delta \underbrace{\sum_{i \in I} w_i \cdot \left<A_i,y\right>}_{=0\textrm{ as }y \in U_3^{(t)}} + 9\alpha^2\delta^2\sum_{i \in I} w_i \cdot \left<A_i,y\right>^2 \Big] \\ 
 &+&\Big[- 2\beta \delta \underbrace{\sum_{i \in I} w_i \cdot \big<A_i^{\circ 2} \circ x,y\big>}_{=0\textrm{ as }y \in U^{(t)}_4} + 9\cdot 4\beta^2\delta^2 \sum_{i \in I} w_i \cdot \left<A_i^{\circ 2} \circ x,y\right>^2 \Big] \\
 &+& \underbrace{\Big[ - \beta \delta^2 \sum_{i \in I} w_i\sum_{j=1}^n y_j^2A_{ij}^2 + 9\beta^2 \delta^4 \sum_{i \in I} w_i \cdot \Big( \sum_{j=1}^n y_j^2A_{ij}^2 \Big)^2 \Big]}_{\leq -\frac{\beta}{2}\delta^2 \sum_{i \in I} w_i \sum_{j=1}^n y_{j}^2A_{ij}^2} \\
 \end{eqnarray*}
Now, we use the fact that $9\beta \delta^2 \sum_{j=1}^n y_j^2 A_{ij}^2 \leq 9\beta \delta^2 \leq \frac{1}{2}$ to get
 \begin{eqnarray*}
\eqref{eq:BeckFialaPotFctBound} &\leq&  \sum_{i \in I} w_i + \delta^2 \underbrace{\sum_{i \in I} w_i \cdot \Big(9\alpha^2 \left<A_i,y\right>^2 - \frac{\beta}{4} \sum_{j=1}^n y_j^2A_{ij}^2\Big)}_{\leq 0\textrm{ since }y \in U_5^{(t)}} \\
&+& \beta \delta^2 \underbrace{\sum_{i \in I} w_i \cdot \Big(36\beta \left<A_i^{\circ 2} \circ x,y\right>^2 - \frac{1}{4} \sum_{j=1}^n y_j^2A_{ij}^2 \Big)}_{\leq 0\textrm{ since }y \in U_6^{(t)}}  \leq \sum_{i \in I} w_i.
\end{eqnarray*}
This proves the claim.
\end{proof}

\subsection{The discrepancy guarantee}

We can now prove that the algorithm indeed finds a vector satisfying the desired 
discrepancy bound: 
\begin{lemma}
For a proper choice of $\alpha := \Theta(\sqrt{\log n})$ and $\beta := \Theta(\log n)$, the algorithm
returns a vector $x := x^{(T)}$ with $\left<A_i,x\right> \leq O(\sqrt{\log n})$ for each row $i \in [m]$.
\end{lemma}
\begin{proof}
First consider a light index $i \in I_{\textrm{light}}$. The potential function never increases, hence 
\[
  e^{\alpha \left<A_i,x\right>} \leq \Phi^{(T)} \leq \Phi^{(0)} \leq m \cdot e^{\beta \cdot C} \leq n^2 \cdot e^{\beta \cdot C}
\]
Taking logarithms and dividing by $\alpha$ gives
\[
 \left<A_i,x\right> \leq \frac{2\log(n)}{\alpha} + \frac{C\beta}{\alpha} \leq (C^2+2)\cdot \sqrt{\log(n)}.
\]
Here the last inequality follows for choices of if  $\alpha := \sqrt{\log(n)}$ and $\beta := C\log(n) = C\alpha^2$.
Now consider a heavy index $i$. Let $t$ be the last iteration when $\sum_{j \in A^{(t)}} A_{ij}^2 > C$. 
Until this point one has $\left<A_i,x^{(t)}\right> = 0$. 
Since $|A_{ij}| \geq 1 / (C^2\sqrt{\log n})$ for every non-zero entry, one has $|\{ j \in A^{(t+1)} : A_{ij} \neq 0\}| \leq C^5\log(n)$. Hence, regardless how those elements are colored, 
one has $\left<A_i,x^{(T)}\right> = \left<A_i,x^{(T)} - x^{(t)}\right> \leq 2\sum_{j \in A^{(t+1)}} |A_{ij}| \leq O(\sqrt{\log n})$.
\end{proof}

\subsection{Quadratic error in subspaces\label{sec:BeckFialaQuadraticErrorInSubspace}}

It remains to prove that the subspaces $U_5^{(t)}$ and $U_6^{(t)}$ used in the algorithm exist
with high enough dimensions. We will prove two lemmas that we keep general: 
\begin{lemma} \label{lem:GeneralSubspaceWithQuadraticInequality}
Let $A,B \in \setR^{m \times n}$ be any matrices with $|B_{ij}| \leq |A_{ij}|$ for all $(i,j) \in [m] \times [n]$
 and let $w_1,\ldots,w_m \geq 0$ be any weights.  
Then for any $k \in \setN$, one can compute a subspace $U$ of dimension at least $\dim(U) \geq (1-\frac{1}{k})n$ in time $O(n^2 (m+n))$ so that 
\[
  \sum_{i=1}^m w_i \cdot (B_iB_i^T \bullet yy^T) \leq k \sum_{i=1}^m w_i \cdot (\mathrm{diag}(A_i^{\circ 2}) \bullet yy^T) \quad \quad \forall y \in U.
\]
\end{lemma}
\begin{proof}
Consider the matrix $L := \sum_{i=1}^m w_i B_iB_i^T$ and $R := k \cdot \sum_{i=1}^m w_i \cdot \textrm{diag}(A_i^{\circ 2})$.
Then the goal is to find a subspace $U$ so that $(L \bullet yy^T) \leq (R \bullet yy^T)$ for all $y \in U$.
First, if we replace $A_i' := \sqrt{w_i} A_i$ and $B_i' := \sqrt{w_i} B_i$, 
then the assumption $|B_{ij}| \leq |A_{ij}|$ is preserved and the claim is not changed. Hence we may assume
that $w_i = 1$ for all $i \in [m]$. If $A^j = \bm{0}$, then also $B^j = \bm{0}$ and $(L \bullet e_je_j^T) = 0 = (R \bullet e_je_j^T)$ which means
that $e_j$ can be added to the subspace. So let us assume that $A^j \neq \bm{0}$ for all $j$.  
Next, if we scale a columns $A^j$ and $B^j$ by some scalar $s$ and 
we scale $y_j$ by $\frac{1}{s}$, then the claim remains invariant. Hence we assume that $\|A^j\|_2 = 1$ for all $j \in [n]$. 
Then 
\[
  \textrm{Tr}[L] = \sum_{i=1}^m \|B_i\|_2^2 = \sum_{j=1}^n \|B^j\|_2^2 \stackrel{|B_{ij}| \leq |A_{ij}|}{\leq} \sum_{j=1}^n \underbrace{\|A^j\|_2^2}_{= 1} = n
\]
On the other hand, 
\[
 R = k\sum_{i=1}^m \textrm{diag}(A_i^{\circ 2}) = k \cdot \textrm{diag}\Big(\Big(\underbrace{\sum_{i=1}^m A_{ij}^2}_{=1}\Big)_{j \in [n]}\Big) = k \cdot I
\]
Then $L$ must have less than $\frac{n}{k}$ eigenvalues of value more than $k$. Then 
we can define $U$ as the span of the eigenvectors of $L$ that have eigenvalue at most $k$.
Computing the matrices $L,R$ takes time $O(mn^2)$ and the eigendecomposition can be done in time $O(n^3)$.
\end{proof}
The existence of the subspace $U_5^{(t)}$ follows from choosing $B := A$ with $k := 16$ and $\frac{\beta}{16\alpha^2} \geq 16$.
The second lemma that we need is the following
\begin{lemma} \label{lem:SubspaceWithQuadraticInequality}
Let $A \in [-\gamma,\gamma]^{m \times n}$ and $x \in [-1,1]^n$. Then
for any $k \in \setN$ one can compute a subspace $U \subseteq \setR^n$ with $\dim(U) \geq n \cdot (1-\frac{1}{k})$ in time $O(n^2(m+n))$
so that 
\[
 \sum_{i=1}^m w_i \cdot \left((A_i^{\circ 2} \circ x)^{\otimes 2} \bullet yy^T\right) \leq k \cdot \gamma^2 \cdot \sum_{i=1}^m w_i \cdot \big(\mathrm{diag}(A_i^{\circ 2})\bullet yy^T\big) \quad\quad \forall y \in U
\]
\end{lemma}
\begin{proof}
We define a matrix $B \in \setR^{m \times n}$ by letting
$B_{i} := \frac{A_i^{\circ 2} \circ x}{\gamma}$. Then $|B_{ij}| \leq |A_{ij}|$ and 
applying Lemma~\ref{lem:GeneralSubspaceWithQuadraticInequality} gives the claim.
\end{proof}
Then applying Lemma~\ref{lem:SubspaceWithQuadraticInequality} with $\gamma := \frac{1}{C\sqrt{\beta}}$ and $k = 16$ guarantees the subspace $U_6^{(t)}$.
For the running time analysis of Theorem~\ref{thm:ConstructiveBeckFialaAlgorithm}, one 
can set $\delta := \Theta(\frac{1}{\sqrt{\log n}})$ and the algorithm only takes $O(n \log (n))$ iterations, 
each taking time $O(n^2 (m+n))$.

\end{document}